\documentclass[final,3p,times]{elsarticle}

\usepackage{latexsym}
\usepackage{amsmath, amssymb, amsthm}
\usepackage{algorithmic,algorithm}
\usepackage{graphics}
\usepackage{enumerate}
\usepackage{color}
\usepackage{caption}
\usepackage{subfigure}

\makeatletter
\def\weitdacherl#1{\mathop{\vbox{\m@th\ialign{##\crcr\noalign{\kern3\p@}%
      $\hfil{\scriptscriptstyle\frown}\hfil$\crcr\noalign{\kern1\p@\nointerlineskip}%
      $\displaystyle{#1}$\crcr}}}}
\def\weithutzl#1{\mathop{\vbox{\m@th\ialign{##\crcr\noalign{\kern3\p@}%
      $\hfil{\scriptscriptstyle\smile}\hfil$\crcr\noalign{\kern1\p@\nointerlineskip}%
      $\displaystyle{#1}$\crcr}}}}
 
\def\part{\@startsection{part}{0}%
  \z@{2.5\linespacing\@plus\linespacing}{1.0\linespacing}%
  {\LARGE\bfseries\raggedright}}
\makeatother

\newtheorem{thm}{Theorem}[section]
\newtheorem{prop}[thm]{Proposition}
\newtheorem{lem}[thm]{Lemma}
\newtheorem{defn}[thm]{Definition}
\newtheorem{cor}[thm]{Corollary}
\newtheorem{rem}{Remark}
\newtheorem{conj}{Conjecture}

\newcommand*{\bigtimes}{\mathop{\hbox{\Large{$\times$}}}}

\newcommand{\la}{\langle}
\newcommand{\ra}{\rangle}

\newcommand\dirmin{\ensuremath{\weithutzl{\times}}\,}
\newcommand\dirmax{\ensuremath{\weitdacherl{\times}}\,}

\newcommand\strmin{\ensuremath{\weithutzl{\boxtimes}}\,}
\newcommand\strmax{\ensuremath{\weitdacherl{\boxtimes}}\,}

\newcommand\NrMin{\ensuremath{|\dirmin|}\,}
\newcommand\NrMax{\ensuremath{|\dirmax|}\,}
\newcommand\skel{\ensuremath{\mathbb{S}}\,}

\renewcommand\simeq{\cong}

\journal{Discrete Applied Mathematics}

\begin{document}

\sloppy

\begin{frontmatter}


\title{Strong Products of Hypergraphs: Unique Prime Factorization Theorems and Algorithms}

\author[SB]{Marc Hellmuth\corref{cor1}}
				\ead{marc.hellmuth@bioinf.uni-sb.de}
				\cortext[cor1]{corresponding author}
\author[SB]{Manuel Noll}
				\ead{mnoll@bioinf.uni-sb.de}
\author[LEI,MIS]{Lydia Ostermeier}
				\ead{glydia@bioinf.uni-leipzig.de}

\address[SB]{Center for Bioinformatics, Saarland University, Building E
             2.1, Room 413, P.O. Box 15 11 50, D-66041 Saarbr\"{u}cken,
             Germany
   					}
\address[LEI]{Bioinformatics Group, Department of Computer Science,
              and Interdisciplinary Center for Bioinformatics,
              University of Leipzig,
              H{\"a}rtelstra{\ss}e 16-18, D-04107 Leipzig, Germany
             }
\address[MIS]{Max Planck Institute for Mathematics in the Sciences,
              Inselstra{\ss}e 22, D-04103 Leipzig, Germany
             }

\begin{abstract}
It is well-known that all finite connected graphs have a unique prime
factor decomposition (PFD) with respect to the strong graph product which can be
computed in polynomial time. Essential for the PFD computation is the 
construction of the so-called Cartesian skeleton of the graphs under
investigation. 

In this contribution, we show that every connected thin hypergraph $H$ has a
unique prime factorization with respect to the normal and strong (hypergraph)
product. Both products coincide with the usual strong \emph{graph} product
whenever $H$ is a graph. We introduce the notion of the Cartesian skeleton
of hypergraphs as a natural generalization of the Cartesian skeleton of
graphs and prove that it is uniquely defined for thin hypergraphs. Moreover,
we show that the Cartesian skeleton of hypergraphs can be determined in
$O(|E|^2)$ time and that the PFD can be computed in $O(|V|^2|E|)$ time, for
hypergraphs $H=(V,E)$ with bounded degree and bounded rank.
\end{abstract}


\end{frontmatter}

\section{Introduction}

As shown by D{\"o}rfler and Imrich \cite{DOeIM-69} and independently by
McKenzie \cite{McK-71}, all finite connected graphs have a unique prime
factor decomposition (PFD) with respect to the strong product. The first
who provided a polynomial-time algorithm for the prime factorization of
strong product graphs were Feigenbaum and Sch{\"a}ffer \cite{FESC-92}. The
latest and fastest approaches are due to Hammack and Imrich \cite{HAIM-09}
and Hellmuth \cite{Hel-11}. In all these approaches, the key idea for the
prime factorization of a strong product graph $G$ is to find a subgraph
$\skel(G)$ of $G$ with special properties, the so-called \emph{Cartesian
skeleton}, that is then decomposed with respect to the \emph{Cartesian}
product. Afterwards, one constructs the prime factors of $G$ using the
information of the PFD of $\skel(G)$. 

Hypergraphs are natural generalizations of graphs, see
\cite{Berge:Hypergraphs}. It is well-known that hypergraphs have a unique
PFD w.r.t. the Cartesian product \cite{Imrich67:Mengensysteme,
OstHellmStad11:CartProd}, which can be computed in polynomial time
\cite{BSV-13}. For more details about hypergraph products, see \cite{HOS-12}. 
As it is shown in \cite{HOS-12}, it is possible to find several non-equivalent
generalizations of the standard graph products to hypergraph products. In
this contribution, we are concerned with two generalizations of the 
strong graph product, namely, the so-called normal product
\cite{Sonntag90:NormalProd} and the strong (hypergraph) product
\cite{HOS-12}. We show that every connected simple thin hypergraph has a
unique PFD with respect to these two products. For this purpose, we
introduce the notion of the Cartesian skeleton of hypergraphs as a
generalization of the Cartesian skeleton of graphs \cite{HAIM-09} and show
that it is uniquely defined for thin hypergraphs. Finally, we give an
algorithm for the computation of the Cartesian skeleton that runs in
$O(|E|^2)$ time and an algorithm for the PFD of hypergraphs that runs in
$O(|V|^2 |E|)$ time, for hypergraphs $H=(V,E)$ with bounded degree and
bounded rank.


\section{Preliminaries}

\subsection{Basic Definitions}

A \emph{hypergraph} $H=(V,E)$ consists of a finite set $V$ and a
collection $E$ of non-empty subsets of $V$. The elements of $V$ are called
\emph{vertices} and the elements of $E$ are called \emph{hyperedges}, or
simply \emph{edges} of the hypergraph. Throughout this contribution, we
only consider hypergraphs without multiple edges and thus, being $E$ a
usual set. If there is a risk of confusion we will denote the vertex set
and the edge set of a hypergraph $H$ explicitly by $V(H)$ and $ E(H)$,
respectively.

Two vertices $u$ and $v$ are \emph{adjacent} in $H=(V, E)$ if there is an
edge $e\in E$ such that $u,v\in e$. The set of all vertices $u$ that are
adjacent to $v$ in $H$ is denoted by $N^H(v)$. The set
$N^H[v]=N^H(v)\cup\{v\}$ is called the \emph{(closed) neighborhood} of $v$.
If any two distinct vertices $u,v\in V$ can be distinguished by their
neighborhoods, that is, $N^H[u]\neq N^H[v]$, then the hypergraph $H=(V, E)$
is called \emph{thin}. A vertex $v$ and an edge $e$ of $H$ are
\emph{incident} if $v\in e$. The \emph{degree} $\deg(v)$ of a vertex $v\in
V$ is the number of edges incident to $v$. The \emph{maximum degree}
$\max_{v\in V} \deg(v)$ is denoted by $\Delta_H$ or just by $\Delta$.

A hypergraph $H=(V,E)$ is \emph{simple} if no edge is contained in any
other edge and $|e|\geq 2$ for all $e\in E$. A hypergraph is \emph{trivial}
if $|V|=1$. The \emph{rank} of a hypergraph $H=(V,E)$ is $r(H)=\max_{e\in
E}|e|$. A hypergraph with $r(H)\leq 2$ is a \emph{graph}.  

A sequence $P_{v_0,v_{k}} = (v_0,e_1,v_1,e_2,\ldots,e_k,v_k)$ in a
hypergraph $H=(V, E)$, where $e_1,\ldots,e_k \in E$ and $v_0,\ldots,v_k\in
V$, such that each $v_{i-1},v_i\in e_i$ for all $i=1,\ldots,k$ and $v_i\neq
v_j$, $e_i\neq e_j$ for all $i\neq j$ with $i,j\in \{1,\dots,k\}$ is called
a \emph{path of length $k$} (joining $v_0$ and $v_k$). The \emph{distance}
$d_H(v,v')$ between two vertices $v,v'$ of $H$ is the length of a
shortest path joining them. A hypergraph $H=(V,E)$ is called
\emph{connected}, if any two distinct vertices are joined by a path.

A \emph{partial hypergraph} $H'=(V',E')$ of a hypergraph $H=(V,E)$, denoted
by $H'\subseteq H$, is a hypergraph such that $V'\subseteq V$ and
$E'\subseteq E$. In the class of graphs partial hypergraphs are called
\emph{subgraphs}. A partial hypergraph $H'\subseteq H$ is a \emph{spanning}
hypergraph of $H$ if $V(H')=V(H)$. $H'\subseteq H$ is \emph{induced} if $E'
= \{e\in E\mid e\subseteq V'\}$. Induced hypergraphs will be denoted by
$\left\la V'\right\ra$. 

For two hypergraphs $H_1=(V_1,E_1)$ and $H_2=(V_2, E_2)$ a
\emph{homomorphism} from $H_1$ into $H_2$ is a mapping $\varphi:
V_1\rightarrow V_2$ such that
$\varphi(e)=\{\varphi(v_1),\ldots,\varphi(v_r)\}$ is an edge in $H_2$, if
$e=\{v_1,\ldots,v_r\}$ is an edge in $H_1$. A homomorphism $\varphi$ that
is bijective is called an \emph{isomorphism} if it holds $\varphi(e)\in
E_2$ if and only if $e\in E_1$. We say, $H_1$ and $H_2$ are
\emph{isomorphic}, in symbols $H_1\cong H_2$, if there exists an
isomorphism between them. If $H_1\cong H_2$ then we will identify their
edge sets and will write for the sake of convenience $E(H_1)=E(H_2)$. An
isomorphism from $H$ into $H$ is called \emph{automorphism}.

A graph $G=(V,E)$ in which all vertices are pairwise adjacent is called
\emph{complete graph} and is denoted by $K_{|V|}$.
The $2$-section $[H]_2$ of a hypergraph $H=(V,E)$ is the graph $(V,E')$
with $E'=\left\{\{x,y\}\subseteq V\mid \,\exists\; e\in E:
  \{x,y\}\subseteq e,\,x\neq y \right\}$, that is, two vertices are adjacent in
$[H]_2$ if they belong to the same hyperedge in $H$, \cite{Berge:Hypergraphs}. 
Thus, every hyperedge of a simple hypergraph $H$ is a complete subgraph in $[H]_2$.

\begin{rem}
	In the sequel of this paper we only consider \emph{finite, simple, connected hypergraphs}, 
	and therefore, call them for the sake of convenience just \emph{hypergraphs}. 
\end{rem}

\subsection{Hypergraph Products}

As shown in \cite{HOS-12}, it is possible to find several non-equivalent
generalizations of the standard graph products to hypergraph products. We
define in the following the Cartesian product $\Box$, the normal product
$\strmin$ and the strong product $\strmax$, where the latter two products
can be considered as generalizations of the usual strong \emph{graph}
product.

In all of these three products, the vertex sets are the Cartesian set
products of the vertex sets of the factors:
$$V(H_1\Box H_2)=V(H_1\strmax H_2) = V(H_1\strmin H_2) = V(H_1)\times V(H_2)$$

For an arbitrary Cartesian set product $V=\bigtimes_{i=1}^n V_i$ of (finitely
many) sets $V_i$, the \emph{projection} $p_j:V\to V_j$ is defined by
$v=(v_1,\dots,v_n) \mapsto v_j$. We will call $v_j$ the \emph{$j$-th
coordinate} of $v\in V$. 
With this notation, the edge sets are defined as
follows.

\begin{center}
\begin{tabular}{llccl}
	Cartesian product: & $e\in E(H_1\Box H_2)$ & if and only if & &
					 $p_i(e)\in E(H_i), p_j(e)\in V(H_j)$ with $i,j\in\{1,2\}$, $i\neq j$. \\
	Strong product: & $e\in E(H_1\strmax H_2)$ & if and only if & $(i)$ & $e\in E(H_1\Box H_2)$  or \\
&	& &		 $(ii)$ & $p_i(e)\in E(H_i)$, for $i=1,2$  and \\& & & & $|e|=\max_{i=1,2} \{|p_i(e)|\}$\\
	Normal product: & $e\in E(H_1\strmin H_2)$ & if and only if & $(i)$ & $e\in E(H_1\Box H_2)$ or\\
	& & & $(ii)$ & $p_i(e)\subseteq e_i \in E(H_i)$, for $i=1,2$  and \\& & & & $|e|=|p_i(e)|=\min_{j=1,2}\{|e_j|\}$
\end{tabular}
\end{center}

For other equivalent definitions, see \cite{HOS-12}. Note, if $H_1$ and
$H_2$ are simple graphs, then the normal and strong (hypergraph) product
coincides with the usual strong \emph{graph} product \cite{Hammack:2011a}. 
The edges, henceforth, of
the normal and the strong product, fulfilling Condition $(i)$ are called
Cartesian edges w.r.t. the factorization $H_1\boxtimes H_2$, and the other
edges are called non-Cartesian w.r.t. $H_1\boxtimes H_2$, $\boxtimes\in
\{\strmin, \strmax\}$, see also Figure \ref{fig:Exmpl}. 

\begin{figure}[tbp]
  \centering
  \subfigure[Shown are the non-Cartesian edges of the normal product $H_1\protect\strmin H_2$.]{
    \label{fig:Labelname1}
    \includegraphics[bb=178 387 579 667, width=0.35\textwidth]{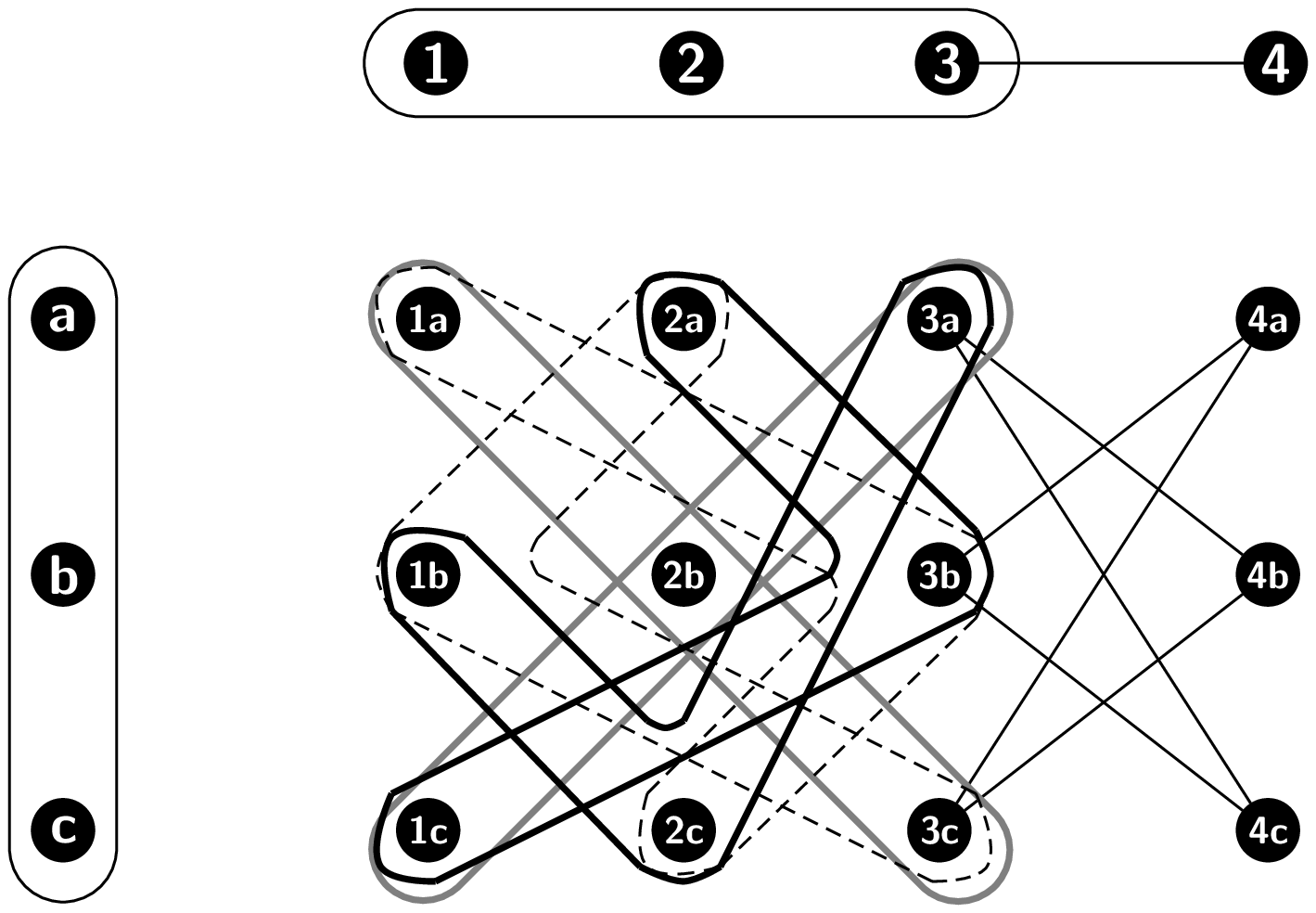}
  }  $\qquad\qquad$ 
  \subfigure[Shown are the non-Cartesian edges of the strong product $H_1\protect\strmax H_2$.]{
    \label{fig:Labelname2}
    \includegraphics[bb=178 387 579 667, width=0.35\textwidth]{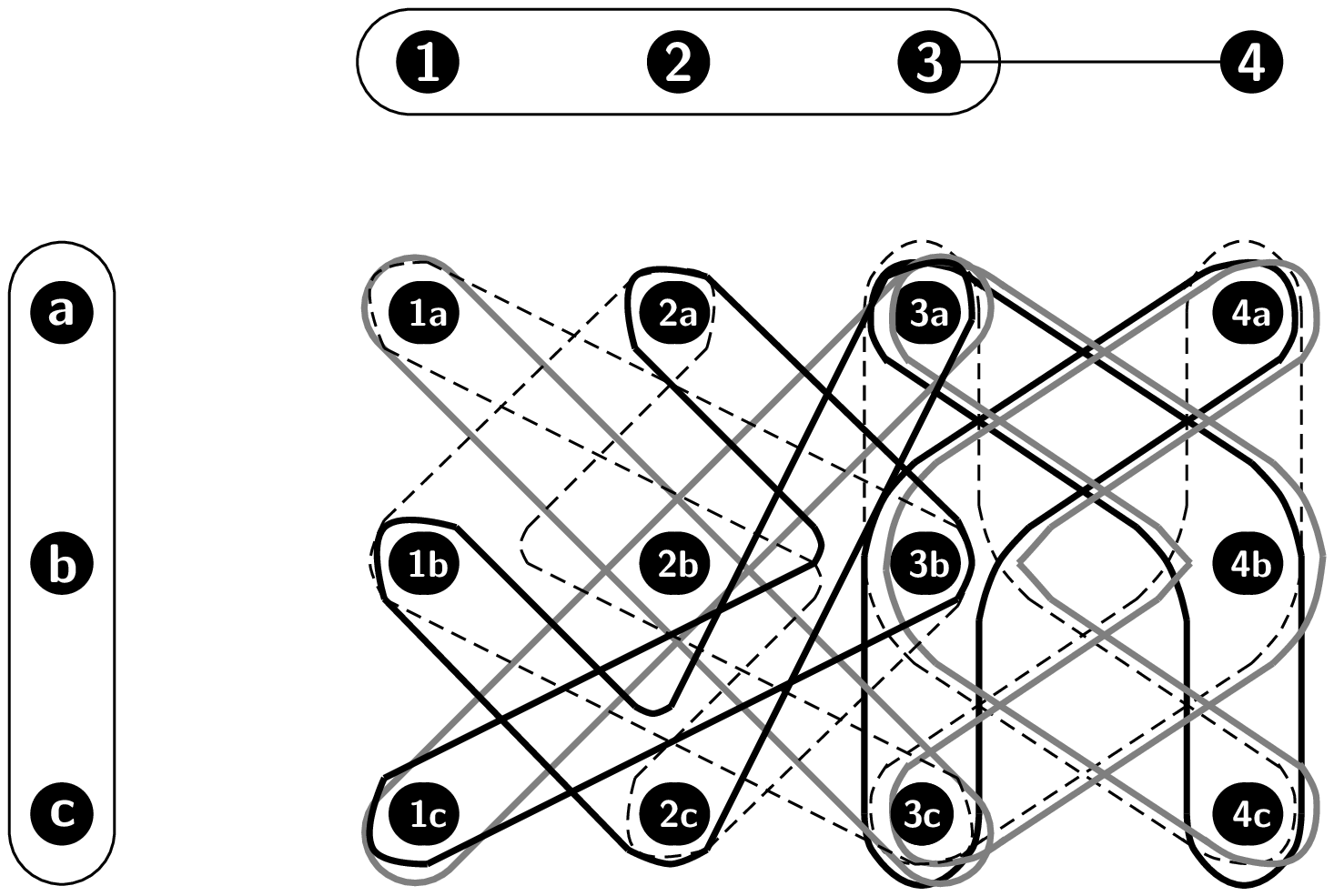}
  }\\
  \centering
  \subfigure[Shown is the Cartesian product $H_1\Box H_2$.]{
    \label{fig:Labelname3}
        \includegraphics[bb=178 387 579 667, width=0.35\textwidth]{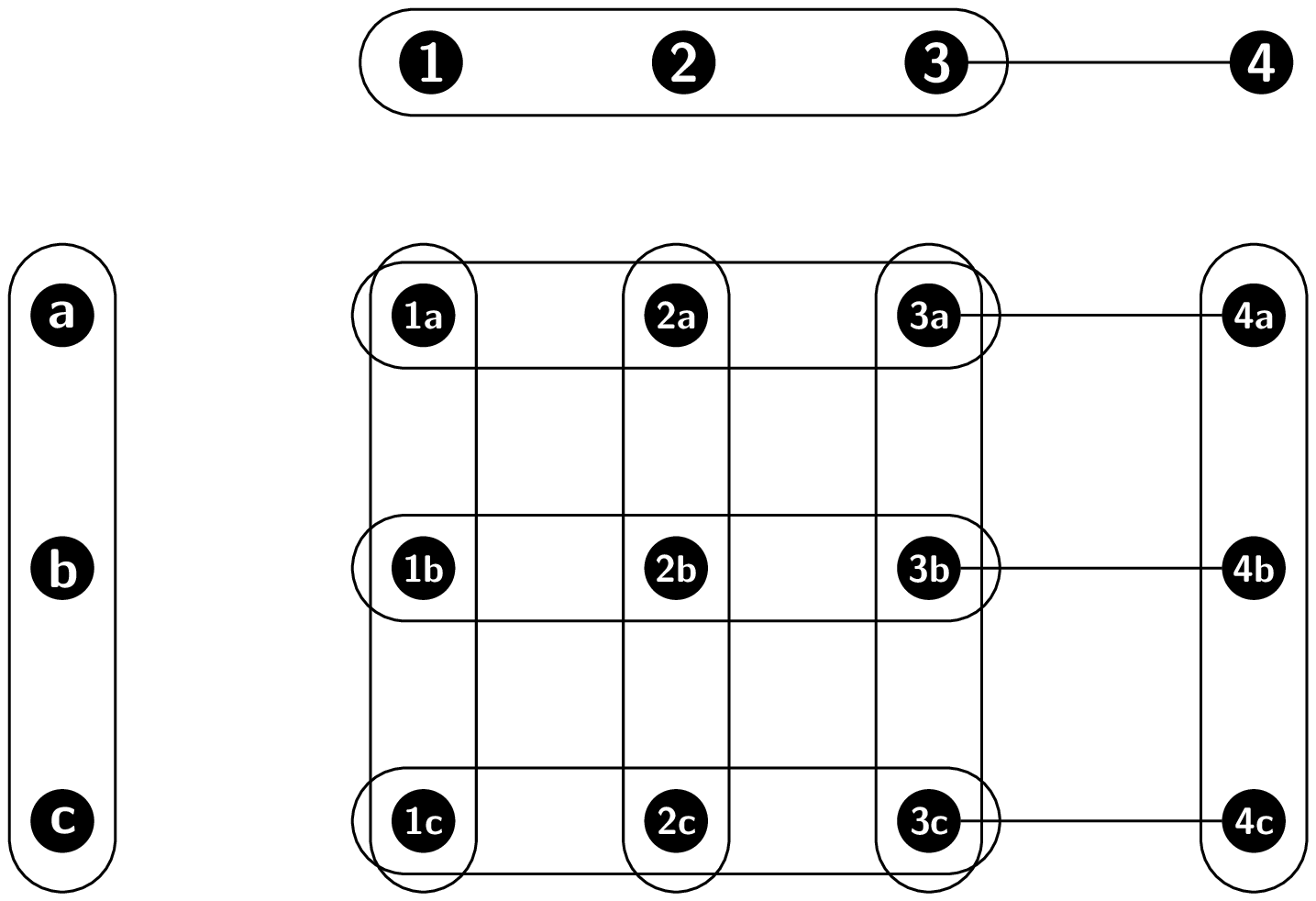}
  } $\qquad\qquad$
  \subfigure[Shown is the 2-section $\ensuremath{[}H_1\ensuremath{]}_2  
	\boxtimes \ensuremath{[}H_2\ensuremath{]}_2 = \ensuremath{[}H_1\protect\strmin H_2\ensuremath{]}_2=
				\ensuremath{[}H_1\protect\strmax H_2\ensuremath{]}_2$]{
    \label{fig:Labelname4}
        \includegraphics[bb=178 387 579 667, width=0.35\textwidth]{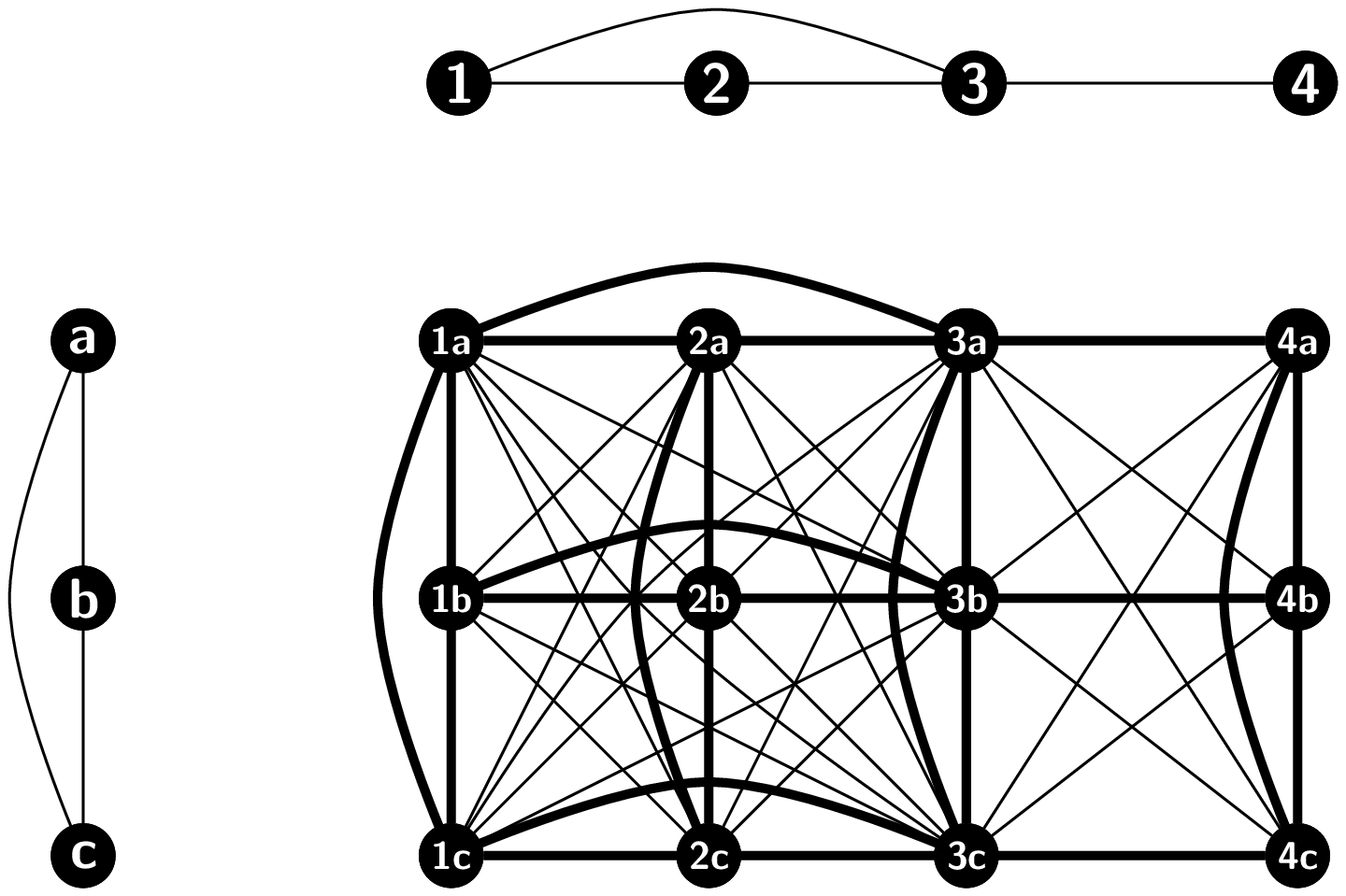}
  }
  \caption{Depicted are the Cartesian and non-Cartesian edges of the different products under 
						investigation. The non-Cartesian edges are drawn in different line-styles, to improve
						visualization. The hypergraph factors $H_1$ and $H_2$ are not thin, 
						and thus neither $H_1\protect\strmin H_2$ nor  $H_1\protect\strmax H_2$ is. }
  \label{fig:Exmpl}
\end{figure}

\begin{rem}
	For the normal product $H=H_1\strmin H_2$ and an edge $e\in E(H)$ holds, 
	if $p_i(e)\subseteq
	e_i\in E(H_i)$ then $|e|\leq |e_i|$. In particular, $p_i(e)\subseteq
	e_i\in E(H_i)$ and $|e|= |e_i|$ implies that $p_i(e)=e_i \in E(H_i)$, $i\in\{1,2\}$.

	For the strong product $H=H_1\strmax H_2$ and an edge $e\in E(H)$ holds, 
	 if $p_i(e)= e_i\in
	E(H_i)$ then $|e|\geq |e_i|$. In particular, $p_i(e)= e_i\in E(H_i)$ and
	$|e|= |e_i|$ implies that $p_i(x)\neq p_i(y)$ for all $x, y\in e$ with
	$x\neq y$, $i\in\{1,2\}$.
\label{REM}
\end{rem}

These three hypergraph products are associative and commutative, thus the
product of finitely many factors is well defined.  
The one-vertex hypergraph $K_1$ without edges
serves as unit element for the Cartesian, normal and strong product, that
is, it holds the \emph{trivial} product representation 
$K_1\circledast H \cong H$, for all $H$ and $\circledast\in\{\Box,
\strmin, \strmax\}$. 
A hypergraph is \emph{prime} w.r.t. $\circledast\in\{\Box, \strmin, \strmax\}$ 
if it has only a trivial product representation. The
Cartesian, normal and strong product of connected hypergraphs is always
connected \cite{HOS-12}. Moreover, it is known that every connected
hypergraph $H=(V,E)$ has a unique prime factor decomposition
w.r.t. (weak) Cartesian product \cite{Imrich67:Mengensysteme,
OstHellmStad11:CartProd}. Furthermore, the number $k$ of Cartesian prime
factors of $H=(V,E)$ is bounded by $\log_2(|V|)$, since every Cartesian
product of $k$ non-trivial hypergraphs has at least $2^k$ vertices. 

Having associativity we can conclude, that a vertex $x$ in these three
products $\circledast_{i\in I} H_i$, $\circledast\in\{\Box, \strmin,
\strmax\}$ is properly ``coordinatized'' by the vector $(x_1, \dots,
x_{|I|})$ whose entries are the vertices $x_i$ of its factors $H_i$. Two
adjacent vertices in the Cartesian product, respectively vertices of a
Cartesian edge in the normal and the strong product, therefore differ in
exactly one coordinate. Moreover, the coordinatization of a product is
equivalent to a (partial) edge coloring of $H$ in which edges $e$ share the
same color $c(e) = k$ if all $x, y\in e$ differ only in the value of a
single coordinate $k$, i.e., if $x_i = y_i$, $i \neq k$ and $x_k \neq y_k$.
This colors the Cartesian edges of $H$ (with respect to the given product
representation). It is easy to see, that for each color $k$ the partial hypergraph
$(V',E')$ with $E' = \{e \in E(H) | c(e) = k\}$ as the set of edges with color
$k$ and  $V'=\cup_{e\in E'} e$ spans $H=(V,E)$, that is, $V' =V$. 

For a given vertex $w\in V(H)$, $H=\circledast_{i\in I} H_i$ the
\emph{$H_j$-layer (through $w$)} is the induced partial hypergraph of $H$   
\[ H_j^w = \left\langle \{v\in V(H) \mid p_k(v)=p_k(w)
  \text{ for } k\neq j \}\right\rangle.\]
For $\circledast\in\{\Box,\strmin,\strmax\}$, we have
$H_j^w\cong H_j$ for all $j\in I,\, w\in V(H)$ \cite{HOS-12}.

Furthermore, for sake of convenience, we introduce the following notations.
Let $H_1$ and $H_2$ be hypergraphs and $\circledast\in\{\Box,\strmin,\strmax\}$.
For $H_1 \circledast H_2$ let $e_i\in E(H_i), i=1,2$
and define 
$$e_1 \circledast e_2 := (e_1,\{e_1\})\circledast (e_2,\{e_2\}).$$
Note, for $\circledast \in \{\Box, \strmin, \strmax\} $ holds $E(e_1
\circledast e_2) \subseteq E(H_1 \circledast H_2)$. Moreover, for an
arbitrary subset $E'\subseteq E(H_1)$ and $x\in V(H_2)$ we denote by
$E'\times \{x\}:= \{e\times \{x\} \mid e\in E'\}$. For later reference we
remark, since $K_1$ is the unit element for $\circledast$ we can rewrite
$E'\times \{x\} = E((V', E') \circledast (x,\emptyset))$ where $V'=\cup_{e
\in E'} e$.

We now give several useful results, that will be needed later on. 

\begin{lem}[\cite{HOS-12}]  
	\label{lem:2sectionStrong}
  The $2$-section of the product $H' \circledast H''$, $\circledast\in\{\Box,\strmin, \strmax\}$ 
	is the respective graph product 
  of the $2$-section of $H'$ and $H''$, more formally:
  $$[H'  \circledast H'']_2 =[H']_2  \circledast [H'']_2.$$
\end{lem}

\begin{lem}[\cite{HOS-12}] 
   \label{lem:simple}
	 The product $H' \circledast H''$, $\circledast\in\{\Box,\strmin, \strmax\}$ 
	  of simple hypergraphs $H'$ and $H''$ is simple.
\end{lem}

\begin{lem}[Distance Formula \cite{HOS-12}]
  \label{lem:distH2}
  Let $H = (V,E)$ be a hypergraph and $x,y\in V$.
  Then the distances between $x$ and $y$
  in $H$ and in $[H]_2$ are the same.
\end{lem}

As for the strong graph product $G=G'\boxtimes G''$ holds that 
$G$ is thin if and only if $G'$ and $G''$ are thin \cite{Hammack:2011a}, 
we obtain together with the latter lemma the following results. 

\begin{cor}
	Let $H=H' \boxtimes H''$, $\boxtimes\in\{\strmax, \strmin\}$. 
	Then it holds $N^H[x] = N^{[H]_2}[x]$. Moreover, 
	$H$ is thin if and only if $[H]_2$ is thin if and only if $H'$ and $H''$
	are thin. 
	\label{cor:thin}
\end{cor}

For later reference we state the next lemma. 

\begin{lem}
  Let $H_1,H_2$ be two hypergraphs.
  For the number $\NrMin$ of non-Cartesian edges 
  in $H=H_1\strmin H_2$ holds
  \[ \NrMin := |E(H_1\strmin H_2)\setminus E(H_1\Box H_2) | =  
	\sum_{e_1\in E_1,e_2\in E_2}\frac{(\max\{|e_1|,|e_2|\})!}{\big||e_1|-|e_2|\big|!}.\]
  For the number $\NrMax$ of non-Cartesian edges 
  in $H=H_1\strmax H_2$ holds
  \[ \NrMax := |E(H_1\strmax H_2)\setminus E(H_1\Box H_2) | =  
	  \sum_{e_1\in E_1,e_2\in E_2}(\min\{|e_1|,|e_2|\})!S_{\max\{|e_1|,|e_2|\},\min\{|e_1|,|e_2|\} },\]
	where $S_{n,k}$ denotes the the Stirling number of the second kind
	$ S_{n,k} = \frac{1}{k!}\sum_{j=0}^k (-1)^{k-j} \binom{k}{j} j^n. 	$
 \label{lem:Nr-nonCart}
\end{lem}
\begin{proof}
 To prove validity of the formula for $\NrMin$, we show that $e$ is a
 non-Cartesian edge in $H_1\strmin H_2$ if and only if there are edges
 $e_1\in E(H_1)$ and $e_2\in E(H_2)$ such that 
 $p_1(x)\mapsto p_2(x)$ for
 all $x\in e$  defines an injective mapping $e_1\to e_2$ whenever
 $|e_1|\leq |e_2|$ and else that $p_2(x)\mapsto p_1(x)$ for all
 $x\in e$ defines an injective mapping $e_2\to e_1$.

 Let $e$ be a non-Cartesian edge in $H_1\strmin H_2$. Clearly, by
 definition of the normal product, there are edges $e_1\in E(H_1)$ and
 $e_2\in E(H_2)$ with $e\in E(e_1\strmin e_2)$. Assume w.l.o.g. $|e_1|\leq
 |e_2|$, otherwise interchange the role of $e_1$ and $e_2$. By definition
 of the normal product it holds $|p_1(e)|=|p_2(e)|=|e|=|e_1|\leq |e_2|$. 
 Thus, we have $p_1(e)=e_1\in E(H_1)$.
 Therefore, 
 we can conclude that all vertices of $e$ differ in each coordinate, and
 thus, $p_1(x)\neq p_1(x')$ implies $p_2(x)\neq p_2(x')$ for all distinct
 vertices $x,x'\in e$. Since $p_2(e)\subseteq e_2$, it follows that
 $p_1(x)\mapsto p_2(x)$, $x\in e$ indeed defines an injective mapping
 $e_1\to e_2$.
 Conversely, if there are edges $e_1\in E(H_1)$ and $e_2\in E(H_2)$ such
 that w.l.o.g. $p_1(x)\mapsto p_2(x)$, $x\in e$ defines an injective
 mapping $e_1\to e_2$, we can conclude that $p_1(e)=e_1$ and
 $p_2(e)\subseteq e_2$. Since $p_1(x)\mapsto p_2(x)$, $x\in e$ is a
 mapping, we have $|e|=|e_1|$ and by injectivity, it follows
 $|e_1|=|p_1(e)|=|p_2(e)|\leq |e_2|$. Hence, $e$ satisfies the condition
 $(ii)$ in the definition of the edges in the normal product and thus, $e\in
 E(H_1\strmin H_2)$.
 Finally, it is well-known, that for any two sets $N$, $M$
 with $|N|\leq |M|$ there are $\frac{|M|!}{(|M|-|N|)!}$ injective mappings 
 from $N$ to $M$. Applying this result to every pair of edges  $e_1\in E(H_1)$ and $e_2\in E(H_2)$
 the assertion for $\NrMin$ follows.

 To prove validity of the formula for $\NrMax$, we show that $e$ is a
 non-Cartesian edge in $H_1\strmax H_2$ if and only if there are edges
 $e_1\in E(H_1)$ and $e_2\in E(H_2)$ such that 
 $p_1(x)\mapsto p_2(x)$ for all $x\in e$
 defines a surjective mapping $e_1\to e_2$ whenever $|e_1|\geq |e_2|$ 
	and else that $p_2(x)\mapsto p_1(x)$ for all $x\in
 e$ defines a surjective mapping $e_2\to e_1$.

  Let $e$ be a non-Cartesian edge in $H_1\strmax H_2$. Clearly, by
 definition of the strong product, there are edges $e_1\in E(H_1)$ and
 $e_2\in E(H_2)$ with $e\in E(e_1\strmax e_2)$. Assume w.l.o.g. $|e_1|\geq
 |e_2|$, otherwise interchange the role of $e_1$ and $e_2$. By definition
 of the strong product it holds that $|e|=|e_1|$ and $p_1(e)=e_1$ which
 implies that $p_1(x)\neq p_1(x')$ for all distinct vertices $x, x'\in e$.
 Thus, $p_1(x)\mapsto p_2(x)$ indeed defines a mapping $e_1\to e_2$. Since
  $p_2(e)=e_2$, this mapping is surjective. 
  Conversely, if there are edges $e_1\in E(H_1)$ and $e_2\in E(H_2)$ such
 that w.l.o.g. $p_1(x)\mapsto p_2(x)$, $x\in e$ defines a surjective
 mapping $e_1\to e_2$ we can conclude that $p_1(e)=e_1$ and $p_2(e)=e_2$
 and thus, in particular that $|p_1(e)|=|e_1|$. Moreover, it follows that
 $|e|=|p_1(e)|$, since $p_1(x)\mapsto p_2(x)$ defines a mapping and
 moreover, $|p_2(e)|\leq |p_1(e)|=|e_1|$, since this mapping is surjective.
 Hence, $e$ satisfies the condition $(ii)$ in the definition of the edges
 in the strong product and thus, $e\in E(H_1\strmax H_2)$.
 Finally, it is well-known, that for any two sets $N$, $M$
 with $|N|\geq |M|$ there are $|M|!S_{|N|,|M|}$ surjective mappings from
 $N$ to $M$. Applying this result to every pair of edges $e_1\in E(H_1)$
 and $e_2\in E(H_2)$ the assertion for $\NrMax$ follows.
\end{proof}

\begin{rem}
	In the sequel of this paper, we will 
	use the symbol $\boxtimes$ for both products, that is, $\boxtimes\in \{\strmin, \strmax\}$,
	unless there is a risk of confusion.
\end{rem}

\section{The Cartesian Skeleton and PFD Uniqueness Results}

\subsection{The Cartesian Skeleton}

For graphs $G$, the key idea of finding the PFD with respect to the strong
product is to find the PFD of a subgraph $\skel(G)$ of $G$, the so-called
\emph{Cartesian skeleton}, with respect to the Cartesian product and
construct the prime factors of $G$ using the information of the PFD of
$\skel(G)$. This concept was first introduced for graphs by Feigenbaum and
Sch{\"a}ffer in \cite{FESC-92} and later on improved by Hammack and Imrich,
see \cite{HAIM-09}. Following the approach of Hammack and Imrich, one
removes edges in $G$ that fulfill so-called dispensability conditions,
resulting in a subgraph $\skel(G)$ that is the desired Cartesian skeleton.
The underlying concept of \emph{dispensability} as defined for graphs in
\cite{HAIM-09} can be generalized in a natural way for hypergraphs. 

\begin{defn}[Dispensability]
An edge $e\in E(H)$ is \emph{dispensable} in $H$
if there exists a vertex $z \in V(H)$ and distinct vertices $x,y \in e$ for which both
of the following statements hold:
\begin{enumerate}
\item $N[x] \cap N[y] \subset N[x] \cap  N[z] \text{ or } N[x] \subset N[z] \subset N[y]$
\item $N[x] \cap N[y] \subset N[y] \cap  N[z] \text{ or } N[y] \subset N[z] \subset N[x]$.
\end{enumerate}
\end{defn}

Note, the latter definition coincides with the one given in \cite{HAIM-09},
if $H$ is a simple graph. Now, we are able to define the Cartesian skeleton for
hypergraphs.

\begin{defn}[Cartesian Skeleton]
Let $D(H)\subseteq E(H)$ be the set of dispensable edges in a given hypergraph $H$.
The \emph{Cartesian skeleton} of a hypergraph $H$ is the
partial hypergraph $\skel[H] \subseteq H$ where all dispensable edges $D(H)$ are
removed from $H$, that is $V(\skel[H]) = V(H)$ and $E(\skel[H])=E(H) \setminus D(H)$.
\end{defn}

In the next theorem, we shortly summarize the results established by
Hammack and Imrich \cite{HAIM-09} concerning the Cartesian skeleton of
graphs and show in the sequel, that these results can easily be transferred
to hypergraphs by usage of its corresponding 2-sections.

\begin{thm}[\cite{HAIM-09}]
	Let $G = G_1 \boxtimes G_2$ be a strong product graph. 
	\begin{enumerate}
	\item If $G$ is thin then every non-dispensable edge $e \in E(G)$ is
	      Cartesian w.r.t. any factorization $G'_1 \boxtimes G_2'$ of $G$. 
	\item	If $G$ is connected, then $\skel(G)$ is connected.
	\item If $G_1$ and $G_2$ are thin graphs then $\skel(G_1\boxtimes G_2) =
	      \skel(G_1) \Box \skel(G_2).$
	\item Any isomorphism $\varphi: G \rightarrow H$, as a map $V (G)
	      \rightarrow V (H)$, is also an isomorphism $\varphi : \skel(G)
	      \rightarrow \skel(H)$.
	\end{enumerate}
	\label{thm:CartSk-hamm}
\end{thm}

Since neighborhoods of vertices in a hypergraph and its 2-section are identical
by  Corollary \ref{cor:thin}  and dispensability is defined only in terms
of neighborhoods, we easily obtain the following lemma and corollary. 

\begin{lem}
	Let $H$ be a hypergraph. The edge $e\in E(H)$ is dispensable in $H$ if
	and only if there is an edge $e'\in E([H]_2)$ with $e'\subseteq e$ and
	$e'$ is dispensable in $[H]_2$.
	\label{lem:disp}
\end{lem}

\begin{cor}
	For all hypergraphs $H$ holds:
		$[\skel(H)]_2= \skel([H]_2)$.
	\label{cor:2-SectionCartSkel}
\end{cor}

From the Distance Formula and Theorem \ref{thm:CartSk-hamm} we obtain immediately:

\begin{cor}
	For all hypergraphs $H$ holds:
	If $H$ is connected then $\skel(H)$ is connected.
	\label{cor:CartSkelConnected}
\end{cor}

\begin{lem}
	Let $H$ be a hypergraph and $H_1\boxtimes H_2$ be an arbitrary
	factorization of $H$. Then it holds that the edge $e$ is Cartesian in $H$
	w.r.t. $H_1\boxtimes H_2$ if and only if $e' $ is Cartesian in $[H_1]_2
	\boxtimes [H_2]_2 = [H]_2$ for all $e'\subseteq e$ with $e'\in E([H]_2)$.
	\label{lem:cart}
\end{lem}
\begin{proof}
	Let $e\in E(H)$ be Cartesian w.r.t. to its factorization $H_1\boxtimes
	H_2$. Then, there is an $i\in \{1,2\}$ with $|p_i(e)| = 1$. Moreover, for
	all $e'\subseteq e$ it holds, $0<|p_i(e')|\leq |p_i(e)|=1$ and hence,
	$|p_i(e')|=1$. Therefore, each edge $e'\in E([H]_2)$ with $e'\subseteq e$
	is Cartesian in $[H]_2 = [H_1]_2 \boxtimes [H_2]_2$.

	By contraposition, assume $e\in E(H)$ is non-Cartesian w.r.t.
	$H_1\boxtimes H_2$. Hence, by definition of the products $\strmin$ and
	$\strmax$ we have $|p_i(e)|>1$, $i=1,2$. Therefore, there are vertices
	$x,y \in e$ with $p_1(x) \neq p_1(y)$. If $p_2(x) \neq p_2(y)$ it follows
	that $e'=\{x,y\}$ is non-Cartesian in $[H]_2 = [H_1]_2 \boxtimes
	[H_2]_2$. If $p_2(x) = p_2(y)$ then there is a vertex $z\in e$ with
	$p_2(x)\neq p_2(z)$. If $p_1(z) \neq p_1(x)$ then the edge $e'=\{x,z\}$
	is non-Cartesian in $[H]_2 = [H_1]_2 \boxtimes [H_2]_2$ and if $p_1(z) =
	p_1(x)\neq p_1(y)$ then the edge $e'=\{y,z\}$ is non-Cartesian in $[H]_2 =
	[H_1]_2 \boxtimes [H_2]_2$.
\end{proof}

\begin{lem}
	Let $H$ be a thin hypergraph. If $e\in E(H)$ is non-dispensable in $H$
	then the edge $e$ is Cartesian w.r.t. any factorization $H_1\boxtimes
	H_2$ of $H$. 
\label{lem:NotDispImpliesCartesian}
\end{lem}
\begin{proof}
 Let $e\in E(H)$ be non-dispensable in $H$. Lemma \ref{lem:disp} implies
 that for all $e'\in E([H]_2)$ with $e'\subseteq e$ holds $e'$ is non-dispensable 
	in $[H]_2$. Furthermore, by Corollary \ref{cor:thin} it holds that $[H]_2$
 is thin. Thus, Theorem \ref{thm:CartSk-hamm} implies that $e'$ is
 Cartesian in $[H]_2$ for all $e'\subseteq e$, which is by Lemma \ref{lem:cart} if and only if $e$
 is Cartesian in $H$. 
\end{proof}

\begin{prop}
	If $H_1$ and $H_2$ are thin hypergraphs, then 
	$\skel(H_1\boxtimes H_2) = \skel(H_1)\Box \skel(H_2)$. 
	\label{prop:cartSkelcart}
\end{prop}
\begin{proof}
	Let $H=H_1\boxtimes H_2$. Lemma  \ref{lem:NotDispImpliesCartesian}
	implies that every non-Cartesian edge is dispensable. Hence we need to
	show, that a Cartesian edge $e \in E(H)$ is dispensable if and only if
	$p_i(e)$ is dispensable whenever $p_i(e)\in E(H_i)$, $i=1,2$. Note,
	exactly for one $i\in \{1,2\}$ holds $p_i(e)\in E(H_i)$ and $p_j(e)\in
	V(H_j)$, $j\neq i$. W.l.o.g. assume $p_1(e)=e_1\in E(H_1)$ and $p_2(e) =
	v_2\in V(H_2)$. 

	Assume that the edge $e$ is dispensable in $H$. Then by
	Lemma~\ref{lem:disp} there exists a dispensable edge $e'\in E([H]_2)$
	with $e'\subseteq e$. Corollary \ref{cor:thin} implies that $[H]_2$ is
	thin and by Theorem \ref{thm:CartSk-hamm} it holds that $\skel([H]_2) =
	\skel([H_1]_2) \Box \skel([H_2]_2)$ and hence, we infer $p_1(e')$ must be
	dispensable in $[H_1]_2$. Since $p_1(e')\subseteq e_1$ and by
	Lemma~\ref{lem:disp}, we conclude that $e_1$ is dispensable in $H_1$.

	Now suppose $e_1$ is dispensable in $H_1$. Again by Lemma~\ref{lem:disp},
	there exists a dispensable edge $e_1'\in E([H_1]_2)$ such that
	$e_1'\subseteq e_1$. Again, by Corollary \ref{cor:thin} it holds that
	$[H]_2$ is thin and Theorem \ref{thm:CartSk-hamm} implies $\skel([H]_2) =
	\skel([H_1]_2) \Box \skel([H_2]_2)$. Therefore, $e'=e_1'\times\{v_2\}$ is
	dispensable in $[H]_2$. By Lemma~\ref{lem:disp} and since $e'\subseteq
	e$, we have $e$ is dispensable in $H$.
\end{proof}

As in \cite{HAIM-09} the Cartesian skeleton $\skel(H)$ 
is defined entirely in terms of the adjacency structure of $H$, and thus, 
we obtain the following immediate consequence of the definition.

\begin{prop}
Any isomorphism $\varphi : H \rightarrow G$, as a map $V(H) \rightarrow V(G)$, 
is also an isomorphism $\varphi: \skel[H] \rightarrow \skel[G].$
\end{prop}

\subsection{Prime Factorization Theorem}

In the following, let $\boxtimes\in \{\strmin, \strmax\}$. Let $A \boxtimes
B$ and $C \boxtimes D$ be two non-trivial decompositions of a simple
connected thin hypergraph $H$. We will show that then $H$ has a finer
factorization of the form $AC\boxtimes AD \boxtimes BC\boxtimes BD$ and $A=
AC\boxtimes AD $, $B=BC\boxtimes BD$, C=$AC\boxtimes BC$ and $D=AD\boxtimes
BD$, see Prop.~\ref{prop:verfeinerung}. Similar as for graphs \cite[page
171-174]{IMKL-00}, this can be used to show that every simple thin
connected hypergraph has a unique prime factorization with respect to the
normal and strong (hypergraph) product. We don't want to conceal the fact,
that in the sequel of this section, we make frequent use of the same
arguments as for graph products in \cite{IMKL-00} and \cite{Hammack:2011a}.

By Proposition \ref{prop:cartSkelcart}, it holds $\skel(H) = \skel(A)\Box
\skel(B) = \skel(C)\Box \skel(D) $. Let $\skel(H)=\Box_{i\in I} H_i$ be the
unique PFD of the Cartesian skeleton of $H$. Hence, the factors $\skel(A)$,
$\skel(B)$, $\skel(C)$ and $\skel(D)$ are all products of or isomorphic to
the Cartesian \emph{prime} factors of $\skel(H)$. Let $I_A$ be the
subset of the index set $I$ with $V(A)=V(\Box_{i \in I_A} H_i)$.
Analogously, the index sets $I_B$, $I_C$ and $I_D$ are defined. 

In the following, we define the hypergraphs $AC, AD, BC$ and $BD$ and as it will turn out
it holds $H\cong AC\boxtimes AD \boxtimes BC\boxtimes BD$.
Therefore, it will be convenient to use only four coordinates
$x=(x_{AC},x_{AD},x_{BC},x_{BD})$ for every vertex $x \in V(H)$. 
With this notation, the projections $p_{AC}:V(H)\rightarrow V(AC)$,
$p_{AD}:V(H)\rightarrow V(AD)$,
$p_{BC}:V(H)\rightarrow V(BC)$,
$p_{BD}:V(H)\rightarrow V(BD)$
are well-defined.

Moreover, the vertex set of $AC$ is defined as 
$V(AC) = V(\Box_{i \in I_A \cap I_C} H_i)$. Analogously,
the vertex sets of $AD$, $BC$ and $BD$ are defined. 
It will be shown that $A=
AC\boxtimes AD $, $B=BC\boxtimes BD$, C=$AC\boxtimes BC$ and $D=AD\boxtimes
BD$. Of course it is possible that not all of the intersections $I_A
\cap I_C, I_A \cap I_D, I_B \cap I_C$ and $I_B \cap I_D$ are nonempty.
Suppose that $I_B \cap I_D = \emptyset$ then $I_A \cap I_D \neq \emptyset$,
since otherwise $I_D=\emptyset$. If in addition $I_A \cap I_C$ were empty,
then $I_A = I_D$ and thus $I_B=I_C$, but then there would be nothing to prove.
Thus, we can assume that all but possibly $I_B \cap I_D$ are nonempty and at
least three of the four coordinates are nontrivial, that is to say, there
are at least two vertices that differ in the first, second and third
coordinates, but it is possible that all vertices have the same fourth
coordinate. 
 
With the definition of the projections $p_A, p_B,p_C $ and $ p_D$
together with the preceding construction of the coordinates
$(x_{AC},x_{AD},x_{BC},x_{BD})$ for vertices $x\in V(H)$, 
we thus have 
$$x_A =p_A(x)=p_A(x_{AC},x_{AD},x_{BC},x_{BD}) = (x_{AC},x_{AD},-,-)=:(x_{AC},x_{AD}) \in V(A), $$
$$x_B =p_B(x)=p_B(x_{AC},x_{AD},x_{BC},x_{BD}) = (-,-,x_{BC},x_{BD})=:(x_{BC},x_{BD}) \in V(B), $$
$$x_C =p_C(x)=p_C(x_{AC},x_{AD},x_{BC},x_{BD})  = (x_{AC},-,x_{BC},-)=:(x_{AC},x_{BC}) \in V(C), $$
$$x_D =p_D(x)=p_D(x_{AC},x_{AD},x_{BC},x_{BD})  = (-,x_{AD},-,x_{BD})=:(x_{AD},x_{BD}) \in V(D). $$

In this way, vertices of $A$, $B$, $C$ and $D$ are coordinatized. 
Thus, the projections 
$p'_{AC}: V(A)\rightarrow V(AC)$ and 
$p''_{AC}: V(C)\rightarrow V(AC)$ 
are well-defined.
Since for all $x\in V(H)$ holds that 
$$p_{AC}(x)=p'_{AC}(p_A(x))=p'_{AC}(x_A)=p''_{AC}(p_C(x))=p''_{AC}(x_C)=x_{AC},$$
we will identify $p_{AC}$ with $p'_{AC}$, resp., $p''_{AC}$, henceforth and simply write $p_{AC}$.
Analogously, we identify the respective projections onto $AD$, $BC$ and $BD$
with $p_{AD}$, $p_{BC}$, $p_{BD}$.

We are now in the position to give the complete definition of the hypergraphs $AC$, $AD$, $BC$ and $BD$.
The vertex set of $AC$ is 
\begin{equation}
\label{eq:VertexSet} 
V(AC) = V(\Box_{i \in I_A \cap I_C} H_i) = \bigtimes_{i \in I_A \cap I_C} V(H_i)
\end{equation}
The edge set of $AC$ is 
\begin{equation}
\label{eq:EdgeSet}
E(AC)=\{e_{AC}\subseteq V(AC) \mid 
\exists e_H \in E(H) \text{ with } p_{AC}(e_H)=e_{AC} \text{ s.t. } \nexists e'_H \in E(H) : p_{AC}(e_H) \subset p_{AC}(e'_H)\} 
\end{equation}
Analogously, the hypergraphs $AD, BC$ and $BD$ are defined.

Equation \eqref{eq:EdgeSet}, that characterizes the edge sets for the
(putative) finer factors $AC, AD, BC$ and $BD$ w.r.t. $\boxtimes$, forces
edges to be maximal with respect to inclusion. We need this definition, in
particular for defining the factors of the normal product, since
projections of edges into the factors might be proper subsets of edges
different from a single vertex.

\begin{rem}
Note, that vertices $x$ are well defined by 
their entries $x_{AC}$, $x_{AD}$, $x_{BC}$ and $x_{BD}$ of their coordinates,
independently from the ordering of $x_{AC}$, $x_{AD}$, $x_{BC}$ and $x_{BD}$, 
since the coordinates will be clearly marked. Therefore, we henceforth distinguish vertices 
just by the entries of their coordinates rather than by the ordering.
\end{rem}

\begin{lem}
\label{lem:max}
Let 
$H \simeq A\boxtimes B \simeq C\boxtimes D$ be a thin hypergraph 
and $AC$ be as defined in 
Equations \eqref{eq:VertexSet} and \eqref{eq:EdgeSet}. 
Then it holds: 
\begin{enumerate}
	\item $e_{AC}\subseteq p_{AC}(e_A)$ implies $e_{AC}=p_{AC}(e_A)$ and 
	      $e_{AC}\subseteq p_{AC}(e_C)$ implies $e_{AC}=p_{AC}(e_C)$
	      for all edges $e_{AC}\in E(AC)$, $e_A\in E(A)$ and $e_C\in E(C)$.
	\item If $p_{AC}(e_H)\in E(AC)$ then $p_A(e_H) \in E(A)$ and $p_C(e_H) \in E(C)$ for every edge $e_H\in E(H)$. 
\end{enumerate}
Analogous results hold for the hypergraphs $AD$, $BC$ and $BD$ with respective edges.
\end{lem}
\begin{proof} 
For the proof of the first statement, let $e_{AC} \in E(AC)$ and assume for
contradiction, that there is an edge $e_A\in E(A)$ with $e_{AC} \subset
p_{AC}(e_A)$. Thus, there is an edge $e_H\in E(H)$ with $e_H=e_A \times
\{x_B\}, x_B\in V(B)$ and therefore, $e_{AC} \subset p_{AC}(e_H)$, which
contradicts the definition of $AC$. Analogously, there is no edge $e_C \in
E(C)$. such that $e_{AC} \subset p_{AC}(e_C)$. 

For the proof of the second statement, let $e_H \in E(H)$ be an arbitrary
edge and assume that $p_{AC}(e_H)\in E(AC)$. Note, if $|p_{AC}(e_H)|>1$
then there are at least two distinct vertices $x, x'\in e_H\in E(H)$ with
$p_{AC}(x)=x_{AC}\neq p_{AC}(x')=x'_{AC}$. Hence, $p_A(x)\neq p_A(x')$ and
$p_C(x)\neq p_C(x')$. Therefore, $|p_{AC}(e_H)|>1$ implies that
$|p_A(e_H)|>1$ and $|p_C(e_H)|>1$ for each edge $e_H \in E(H)$. Thus,
whenever $p_{AC}(e_H) \in E(AC)$ then the projections $p_A(e_H)$ and
$p_C(e_H)$ cannot be a single vertex.

If $\boxtimes =\strmax$ then the condition $p_A(e_H) \in E(A)$ and
$p_C(e_H) \in E(C)$ is trivially fulfilled by the definition of $\strmax$,
since $p_{AC}(e_H)\in E(AC)$ and thus, $|p_{AC}(e_H)|>1$. 

Now, consider the product $\strmin$. Note, since $e_H\in E(e_A\strmin e_B)$
for some $e_A\in E(A)$, $e_B\in E(B)$ we can conclude by definition of the
normal product that $p_A(e_H)\subseteq e_A$ and thus,
$p_{AC}(e_H)=p_{AC}(p_A(e_H))\subseteq p_{AC}(e_A)$. By assumption, 
we have $p_{AC}(e_H)\in E(AC)$ and therefore, Item $(1)$ of this
lemma implies that $p_{AC}(e_H)= p_{AC}(e_A)$. Moreover, it holds that
$|e_H|\geq |p_{AC}(e_H)|$ and by Remark \ref{REM} we have $|e_A|\geq
|e_H|\geq |p_{AC}(e_H)|$. 
Since $H\cong A\strmin B$ there is an edge $e'_H = e_A \times \{x_B\} \in
E(H)$ which implies that $p_C(e'_H) = p_{AC}(e_A) \times \{x_{BC}\}$. Thus,
$|p_C(e'_H)| = |p_{AC}(e_A)|\leq|e_A|=|e'_H|$, since $e'_H$ is Cartesian
w.r.t. $A\strmin B$.
Since $H\simeq C\strmin D$ and by the definition of the normal product it
holds $|p_C(e_H')|=|e_H'|$, and therefore, $|e_A|=|p_{AC}(e_A)|=|p_{AC}(e_H)|$. 
Since $|e_A|\geq |e_H|\geq |p_{AC}(e_H)|$ it holds $|e_H|=|e_A|$. 
Thus, we can conclude by Remark~\ref{REM} that $p_A(e_H) \in E(A)$.
By similar arguments one can show that $p_C(e_H) \in E(C)$. 
\end{proof}

\begin{lem}
Let 
$H \simeq A\boxtimes B \simeq C\boxtimes D$ be a thin hypergraph 
and $AC$ and $BC$ be as defined in Equations \eqref{eq:VertexSet} and \eqref{eq:EdgeSet}. 
Then for all $e_{AC} \in E(AC)$ and all $x_{BC} \in V(BC)$ there is an
edge $e_C = e_{AC} \times \{x_{BC}\} \in E(C)$.
Analogous results hold for the hypergraphs $AD$, $BC$ and $BD$ with respective edges.
\label{lem:eC}
\end{lem}
\begin{proof}
	Let $e_{AC} \in E(AC)$ be an arbitrary edge. By definition of $AC$, 
	there is an edge $e_H \in E(H)$ with $p_{AC}(e_H)=e_{AC}$. 
	Note, by the same arguments as in the proof of Lemma \ref{lem:max}
	it holds that 
	$|p_{AC}(e_H)|>1$ implies $|p_A(e_H)|>1$ and $|p_C(e_H)|>1$ 
	for each $e_H \in E(H)$. 

	Since $e_H\in E(A\boxtimes B)$, there is an edge $e_A\in E(A)$ s.t.
	$p_A(e_H)\subseteq e_A$. Therefore, $e_{AC} = p_{AC}(e_H) =
	p_{AC}(p_A(e_H))\subseteq p_{AC}(e_A)$ which implies together with Lemma
	\ref{lem:max} (1), that $p_{AC}(e_A) = e_{AC}$. By Lemma \ref{lem:max} (2),
	we have $p_A(e_H) = e_A$. 
	Therefore, there is an
	edge of the form $e_A\times \{x_B\}\in E(H)$. W.l.o.g. let us assume that
	$e_H$ is chosen s.t. $e_H=e_A\times \{x_B\}$. Since we also have $e_H \in
	E(C\boxtimes D)$ there is an edge $e_C\in E(C)$ s.t. $p_C(e_H) \subseteq
	e_C$. 
	Analogously, we can conclude by Lemma~\ref{lem:max} $p_C(e_H) = e_C$.
	Hence, $e_C = p_{AC}(e_A) \times \{x_{BC}\}=e_{AC} \times \{x_{BC}\} \in
	E(C)$.
\end{proof}

\begin{lem}
Let 
$H \simeq A\boxtimes B \simeq C\boxtimes D$ be a thin hypergraph 
and $AC$ and $BC$ be as defined in Equation
\eqref{eq:VertexSet} and \eqref{eq:EdgeSet}. Then it holds that
$p_{AC}(e_C) \in E(AC)$ for all edges $e_C\in E(C)$ with
$e_C=p_{AC}(e_C)\times \{x_{BC}\}$, $x_{BC}\in V(BC)$.
Analogous results hold for the hypergraphs $AD$, $BC$ and $BD$ with respective edges.
	\label{lem:eC2}
\end{lem}
\begin{proof}
	Let $e_C=p_{AC}(e_C)\times \{x_{BC}\}\in E(C)$. Since $H\cong C\boxtimes
	D$, there is an edge $e_H = e_C \times \{x_D\} \in E(H)$. It holds
	$p_{AC}(e_C) = p_{AC}(p_C(e_H))=p_{AC}(e_H) \subseteq e_{AC}\in E(AC)$.
	Assume for contradiction, that $p_{AC}(e_C) \subset e_{AC}\in E(AC)$.
	Then there is by definition of $AC$ another edge $e'_H \in E(H)$ with
	$p_{AC}(e'_H)=e_{AC}$. Since $H\cong A\boxtimes B$, there is an edge
	$e_A\in E(A)$ with $p_A(e'_H)\subseteq e_A$. Hence, we have $p_{AC}(e_H)
	= p_{AC}(e_C) \subset p_{AC}(e'_H)=p_{AC}(p_A(e'_H))\subseteq
	p_{AC}(e_A)$, shortly, $p_{AC}(e_H)\subset p_{AC}(e_A)$. By definition of
	the normal and the strong product, there is an edge $e''_H = e_A\times
	\{x_B\}\in E(H)$. Since we assumed to have $e_C=p_{AC}(e_C)\times
	\{x_{BC}\}$ it holds $e_C\subset e_{AC}\times
	\{x_{BC}\}=p_C(e''_H)\subseteq e'_C$ for some $e'_C\in E(C)$
	contradicting that $C$ is simple. Thus, $p_{AC}(e_C) = e_{AC} \in E(AC)$.
\end{proof}

\begin{cor}
Let 
$H \simeq A\boxtimes B \simeq C\boxtimes D$ be a thin hypergraph 
and $AC$,$AD$,$BC$ and $BD$ be as defined in Equations \eqref{eq:VertexSet} and \eqref{eq:EdgeSet}. 
Then it holds that $e_{AC} \in E(AC)$  if and only if there is an edge
$e_H \in E(H)$ with $e_H = e_{AC}\times \{x_{AD}\}\times \{x_{BC}\}\times \{x_{BD}\}$, 
$x_{AD}\in V(AD)$, $x_{BC}\in V(BC)$,$x_{BD}\in V(BD)$.
Analogous results hold for respective edges of the hypergraphs $AD$, $BC$ and $BD$.
	\label{cor:eCeC2}
\end{cor}
\begin{proof}
	If $e_{AC}\in E(AC)$ then by Lemma~\ref{lem:eC} there is an edge $e_C =
	e_{AC}\times \{x_{BC}\}\in E(C)$. Since $H\cong C\boxtimes D$ and by
	choice of the coordinates, there is an edge $e_H = e_C \times \{x_D\} \in
	E(H)$ with $x_D=(x_{AD}, x_{BD})$. Hence, $e_H$ can be written as 
	$e_{AC}\times \{x_{AD}\}\times \{x_{BC}\}\times \{x_{BD}\}$.  

	If $e_H = e_{AC}\times \{x_{AD}\}\times \{x_{BC}\}\times \{x_{BD}\}$ it
	follows that $|p_B(e_H)| = 1$ and $|p_D(e_H)| = 1$ and thus, this edge
	$e_H$ is Cartesian in $A\boxtimes B$ and $C\boxtimes D$. Therefore,
	$p_A(e_H) \in E(A)$ and $p_C(e_H)\in E(C)$. Now, suppose for
	contradiction that $e_{AC} \not\in E(AC)$. By definition of $AC$, there
	is an edge $e'_H$ with $p_{AC}(e'_H)\in E(AC)$ such that
	$e_{AC}=p_{AC}(e_H)\subset p_{AC}(e'_H)$. By Lemma \ref{lem:eC} there is
	an edge $e_C = p_{AC}(e'_H)\times \{x_{BC}\}$ and hence, an edge $e''_H =
	p_{AC}(e'_H)\times \{x_{AD}\}\times \{x_{BC}\}\times \{x_{BD}\}$, which
	implies that $e_H \subset e''_H$, contradicting that $H$ is simple. 
\end{proof}

\begin{lem}
	Let 
	$H \simeq A\boxtimes B \simeq
	C\boxtimes D$ be a thin hypergraph and $AC$, $AD$, $BC$ and $BD$ be as defined in Equations
	\eqref{eq:VertexSet} and \eqref{eq:EdgeSet}. Then for all $e_{AC}\in
	E(AC)$, $e_{AD}\in E(AD)$ and $x_B\in V(B)$ it holds that
	$E(e_{AC}\boxtimes e_{AD}) \times \{x_B\} \subseteq E(H)$. Analogous
	results hold with respective edges in the hypergraphs $BC$ and $BD$ and
	vertices $x_A\in V(A),x_C\in V(C)$ and $x_D\in V(D)$.
	\label{lem:edgeProducts}
\end{lem}
\begin{proof}
	Let $x_B = (x_{BC}, x_{BD})\in V(B)$ with $x_{BC}\in V(BC)$, $x_{BD}\in
	V(BD)$, $e_{AC}\in E(AC)$ and $e_{AD}\in E(AD)$. By Lemma \ref{lem:eC}
	there is an edge $e_C = e_{AC} \times \{x_{BC}\} \in E(C)$ and
	analogously, there is also an edge $e_D = e_{AD} \times \{x_{BD}\} \in
	E(D)$. Hence, it holds: $E(e_{AC}\boxtimes e_{AD})\times\{x_B\}=
	E(e_{AC}\boxtimes e_{AD}\boxtimes
	(\{x_{BC}\},\emptyset)\boxtimes(\{x_{BD}\},\emptyset)) =E((e_{AC}\times
	\{x_{BC}\})\boxtimes(e_{AC}\times \{x_{BD}\})) =E(e_C\boxtimes
	e_D)\subseteq E(H)$. 
\end{proof}

\begin{lem}
	Let 
	$H \simeq A\boxtimes B \simeq
C\boxtimes D$ be a thin hypergraph and $AC$ and $AD$ be as defined in Equations
\eqref{eq:VertexSet} and \eqref{eq:EdgeSet}. Then for all edges $e_{A}\in
E(A)$ there is an edge $e_{AC}\in E(AC)$ and $e_{AD}\in E(AD)$ such that
$e_{A}\in E(e_{AC}\boxtimes e_{AD})$. Analogous results hold for the
hypergraphs $B$, $C$, $D$ with respective edges from $AC$, $AD$, $BC$ and
$BD$, whenever $ I_B\cap I_D\neq \emptyset$.
	\label{lem:edgeProducts2}
\end{lem}
\begin{proof}
	Let $e_{A}\in E(A)$ and $x_B=(x_{BC},x_{BD})\in V(B)$. Since $H \simeq
	A\boxtimes B$, there is a Cartesian edge $e_H = e_A \times \{x_B\}\in
	E(H)$. Furthermore, since $H \simeq C\boxtimes D$ and by definition of
	the normal and the strong product, we can conclude that 
	$p_C(e_H)\in V(C)$ or there is an edge
	$e_C\in E(C)$ with $p_C(e_H)\subseteq e_C$, as well as, $p_D(e_H)\in
	V(D)$ or there is an edge $e_D\in E(D)$ with $p_D(e_H)\subseteq e_D $. 

  Assume first $x_D=p_D(e_H)\in V(D)$. Then $p_C(e_H)=e_C \in E(C)$, that is,
	$e_H = e_C\times \{x_D\}$. 
	Note, coordinates of vertices $x_C \in e_C$ are given by $(x_{AC},
	x_{BC})$. Since $e_H = e_A \times \{x_B\}\in E(H)$ it holds that
	$p_{BC}(e_C) = p_{BC}(p_C(e_H))=p_{BC}(e_H)=x_{BC}$. Therefore, $e_H$ can
	be written as $ p_{AC}(e_C) \times\{x_{BC}\}\times \{x_{D}\}$. Moreover,
	$p_{AC}(e_C) = p_{AC}(e_H)=p_{AC}(e_A)$ and hence, $p_C(e_H) = e_C =
	p_{AC}(e_A) \times \{x_{BC}\}\in E(C)$.
	Now, Lemma \ref{lem:eC2} implies that $p_{AC}(e_A)=e_{AC}\in E(AC)$.
	Moreover, it holds
	$p_{AD}(e_A)=p_{AD}(e_H)=p_{AD}(p_D(e_H))=p_{AD}(x_D)=x_{AD}\in V(AD)$
	and therefore, $e_A= e_{AC}\times \{x_{AD}\}$ and thus, $e_A\in
	E(e_{AC}\boxtimes e_{AD})$ for all $e_{AD}$ with $\{x_{AD}\}\in e_{AD}$.
	Analogously, we infer that $e_A = \{x_{AC}\}\times e_{AD}$, $x_{AC}\in
	V(AC)$ and therefore, $e_A\in E(e_{AC}\boxtimes e_{AD})$ for all $e_{AC}$
	with $x_{AC}\in e_{AC}$ if $p_C(e_H)\in V(C)$.

	Now, we treat the case $p_C(e_H)\subseteq e_C \in E(C)$ and
	$p_D(e_H)\subseteq e_D \in E(D)$ and consider the different products
	$\strmin$ and $\strmax$ separately.

	In case $\strmax$ we have, $p_C(e_H) = e_C=p_{AC}(e_H)\times\{x_{BC}\}
	\in E(C)$ and $p_D(e_H) = e_D = p_{AD}(e_H)\times\{x_{BD}\}\in E(D)$ and
	by the same arguments as before, $p_{AC}(e_H)=p_{AC}(e_A)=e_{AC}\in
	E(AC)$ and $p_{AD}(e_H) = p_{AD}(e_A)=e_{AD}\in E(AD)$. Since
	$e_H=e_A\times\{x_B\}\in E(e_C\strmax e_D)$ and $E(e_C\strmax e_D) =
	E(e_{AC}\strmax e_{AD} )\times \{x_B\}$ we can conclude that $e_A \in
	E(e_{AC}\strmax e_{AD})$.

	In case $\strmin$ we have, $p_{AC}(e_A) = p_{AC}(e_H) =
	p_{AC}(p_C(e_H))\subseteq p_{AC}(e_c) = p_{AC}(e_C \times \{x_D\})$ with
	$e_C \times \{x_D\} \in E(H)$ and therefore $p_{AC}(e_A)\subseteq
	p_{AC}(e_C \times \{x_D\}) \subseteq e_{AC} \in E(AC)$. Analogously it
	holds $p_{AD}(e_A) \subseteq e_{AD}\in E(AD)$. Note, by definition of
	$\strmin$ it holds $p_C(e_H) = e_C$ or $p_D(e_H) = e_D$. Lemma
	\ref{lem:eC2} implies that if $p_C(e_H) = e_C$ then $p_{AC}(e_A)=e_{AC}$
	and if $p_D(e_H) = e_D$ then $p_{AD}(e_A)=e_{AD}$. Furthermore, it holds
	by definition of the normal product $|p_C(e_H)| = |p_D(e_H)|$. If
	$p_C(e_H) = e_C$ then, by the choice of $e_H$, we have
	$|e_{AC}|=|e_{C}|=|p_C(e_{H})|=|p_D(e_{H})|=|p_{AD}(e_A)|\leq |e_{AD}|$.
	If $p_D(e_H) = e_D$ we have
	$|e_{AD}|=|e_{D}|=|p_D(e_{H})|=|p_C(e_{H})|=|p_{AC}(e_A)|\leq |e_{AC}|$.
	Therefore, we can conclude that $|e_A| = |e_H| = \min\{|e_C|,
	|e_D|\}=\min\{|e_{AC}|, |e_{AD}|\}$ and thus, $e_A\in E(e_{AC}\strmin
	e_{AD})$. 
\end{proof}

\begin{prop}
	Let 
	$H \simeq A\boxtimes B \simeq C\boxtimes D$ be a thin hypergraph. 
	Then there exists a decomposition 
	$$H\cong AC \boxtimes AD \boxtimes BC \boxtimes BD$$	of $H$ such that
	$A=AC \boxtimes AD$, $B=BC \boxtimes BD$, $C= AC \boxtimes BC$ and 	
	$D = AD \boxtimes BD$.
	\label{prop:verfeinerung}
\end{prop}
\begin{proof}
	First we show that there is a decomposition $AC \boxtimes AD$  of $A$. 
	Let $AC$ and $AD$ be defined as in Equation \eqref{eq:VertexSet} and \eqref{eq:EdgeSet}. 
	Thus, by construction of $AC$ and $AD$ we have $V(A) = V(AC) \times V(AD)$. 
	Therefore, we need to show that $E(A) = E(AC\boxtimes AD)$. 

	By Lemma \ref{lem:edgeProducts2} and since 
	$E(e_{AC} \boxtimes e_{AD}) \subseteq E(AC \boxtimes AD)$ for all $e_{AC}\in E(AC)$
	and $e_{AD}\in E(AD)$	 we have $E(A) \subseteq E(AC\boxtimes AD)$. 
	
	Let $e\in E(AC\boxtimes AD)$. Hence, there is an edge $e_{AC}\in E(AC)$ and
	$e_{AD}\in E(AD)$ with $e\in E(e_{AC}\boxtimes e_{AD})$. 
	By Lemma \ref{lem:edgeProducts} we can conclude that there is a vertex 
	$x_B\in V(B)$ such that 
	$e \times \{x_B\} \in E(e_{AC}\boxtimes e_{AD}) \times \{x_B\} \subseteq E(H)$. 
	Since $e = p_A(e \times \{x_B \})\in E(A)$, the statement follows. 
	
	By analogous arguments one shows that the results hold also for 
	$B$, $C$ and $D$, whenever $I_B \cap I_D\neq \emptyset$.
	If $I_B \cap I_D =  \emptyset$ then we can conclude that
	$I_B = (I_C\cap I_B) \cup (I_D\cap I_B) = I_C\cap I_B$ and 
	$I_D = (I_A\cap I_D) \cup (I_B\cap I_D) = I_A\cap I_D$.
	Hence, by definition of the vertex sets 
	$V(BC)$ and $V(AD)$ together with 
	Lemma \ref{lem:eC} and \ref{lem:eC2} we obtain that
	$B\cong BC$ and $D\cong AD$ and thus, the assertion follows.
\end{proof}

\begin{thm}
	Connected, thin hypergraphs have a unique prime factor decomposition
	with respect to the normal product $\strmin$ and the strong product 
	$\strmax$, up to isomorphism and the order	of the factors. 
	\label{thm:upfd}
\end{thm}
\begin{proof}
	Reasoning exactly as in the proof for graphs in \cite[Lemma 5.38]{IMKL-00}, 
	and by usage of Prop.~\ref{prop:verfeinerung} we obtain the desired result.
\end{proof}

We conclude this section by discussing the term ``thinness''. It is
well-known that, although the PFD for a given graph $G$ w.r.t. the strong
graph product is unique, the coordinatizations might not be
\cite{Hammack:2011a}. Therefore, the assignment of an edge being Cartesian
or non-Cartesian is not unique in general. The reason for the non-unique
coordinatizations is the existence of automorphisms that interchange
vertices $u$ and $v$, which is possible whenever $u$ and $v$ have the same
neighborhoods and thus, if $G$ is not thin. Thus, an important issue in the
context of strong graph products is whether or not two vertices can be
distinguished by their neighborhoods. The same holds for the normal and
strong hypergraph product, as well. For graphs $G=(V,E)$, one defines the equivalence
relation $S$ on $V$ with $uSv$ iff $N^G[u] = N^G[v]$ and computes a so-called
quotient graph $G/S$ which is a thin graph. For this graph $G/S$ the PFD is
computed and one uses afterwards the knowledge of the \emph{cardinalities}
of the S-classes \emph{only}, to find the prime factors of $G$. For graphs, one
profits from the fact that all vertices $u_1, \dots,u_n \in V(G)$ that
share the same neighborhoods induce a complete subgraph $K_n$. Even in the
proofs for the uniqueness results for the PFD of the strong graph product
of non-thin graphs, this fact is utilized. However, this technique cannot be used for
hypergraphs in general, as the partial hypergraph formed by vertices that
share the same neighborhoods need not to be isomorphic, although the
cardinalities of the S-classes might be the same. So far, we do not know,
how to resolve this problem and state the following conjecture. 

\begin{conj}
	Connected, simple, non-thin hypergraphs have a unique prime factor decomposition
	w.r.t. $\strmin$ and $\strmax$, up to isomorphism and the order
	of the factors. 
	\label{conj:upfd}
\end{conj}

\section{Algorithms for the Construction of the Cartesian Skeleton and the Prime Factors}

As shown by Bretto et al. \cite{BSV-13} the PFD of hypergraphs 
with respect to the Cartesian product can be computed in 
polynomial time.

\begin{thm}[\cite{BSV-13}]
The prime factors w.r.t. the Cartesian product of a given 
connected simple hypergraph $H=(V,E)$ with maximum degree $\Delta$ and  rank $r$
can be computed in $O(|V||E|\Delta^6 r^6)$, that is, in $O(|V||E|)$ time 
for hypergraphs $H$ with a bounded rank and a bounded degree. 
\label{thm:bretto}
\end{thm}

The algorithm for computing the PFD of a given hypergraph with respect to
the normal and the strong product works as follows. Analogously as for
graphs, the key idea of finding the PFD with respect to $\boxtimes\in
\{\strmin, \strmax\}$ is to find the PFD of its Cartesian skeleton
$\skel(H)$ with respect to the Cartesian product and to construct the prime
factors of $H$ using the information of the PFD of $\skel(H)$. In Algorithm
\ref{alg:CartSk} the pseudocode for determining the Cartesian skeleton
$\skel(H)$ is given. This Cartesian skeleton is afterwards factorized with
the Algorithm of \texttt{Bretto et al.} \cite{BSV-13} and one obtains the
Cartesian prime factors of $\skel(H)$. Note, for an arbitrary factorization
$H=H_1\boxtimes H_2$ of a thin hypergraph $H$, Proposition
\ref{prop:cartSkelcart} asserts that $\skel(H_1\boxtimes H_2) =
\skel(H_1)\Box \skel(H_2)$. Since $\skel(H_i)$ is a spanning hypergraph of
$H_i$, $i=1,2$, it follows that the $\skel(H_i)$-layers of $\skel(H_1)\Box
\skel(H_2)$ have the same vertex sets as the $H_i$-layers of $H_1\boxtimes
H_2$. Moreover, if $\boxtimes_{i\in I} H_i$ is the unique PFD of $H$ then
we have $\skel(H)=\Box_{i\in I} \skel(H_i)$. Since $\skel(H_i)$, $i\in I$
need not to be prime with respect to the Cartesian product, we can infer
that the number of Cartesian prime factors of $\skel(H)$, can be larger
than the number of the strong or normal prime factors. Hence, given the PFD of
$\skel(H)$ it might be necessary to combine several Cartesian factors to
get the strong or normal prime factors of $H$. These steps for computing the PFD with
respect to $\boxtimes\in \{\strmin, \strmax\}$ of a thin hypergraph are
summarized in Algorithm \ref{alg:PFD}.

For proving the time complexity of Algorithm \ref{alg:CartSk} we need the following
appealing result, established by Hammack and Imrich. 

\begin{lem}[\cite{HAIM-09}]
	For a given graph $G=(V,E)$ with maximum degree $\Delta$
	the set of dispensable edges $D(H)$ and in particular, the
	 Cartesian skeleton $\skel(G)$ can be computed in 
	$O(\min\{|E|^2, |E|\Delta^2\})$ time.
	\label{lem:timeCartSk-graph}
\end{lem}

\begin{lem}
	For a given hypergraph $H=(V,E)$ with maximum degree $\Delta$ and rank $r$, 
	Algorithm \ref{alg:CartSk} computes the Cartesian skeleton $\skel(H)$ in 
	$O(|E|^2 r^4)$ time.
	\label{lem:timeCartSk-hypergraph}
\end{lem}
\begin{proof}
	The correctness of the algorithm follows immediately from Lemma \ref{lem:disp}.

 	For the time complexity observe that $[H]_2$ has at most
 $|E|\binom{r}{2}$ edges and that the maximum degree of $[H]_2$ is at most
 $\Delta (r-1)$. Hence, Lemma \ref{lem:timeCartSk-graph} implies that the
 computation of the set $D([H]_2)$ takes $O(\min \{|E|^2 r^4, |E|r^2
 \Delta^2 r^2\}) = O(|E|^2 r^4)$ time. To check whether one of the at most
 $O(|E|r^2)$ pairs $\{x, y\}\in D([H]_2)$ is contained in one of the $|E|$
 edges in $H$ we need $O(|E|^2 r^2)$ time, from which we can conclude the
 statement. 
\end{proof}

For computing the time complexity of Algorithm \ref{alg:PFD} we first need the
following lemma. 

\begin{lem}
	Let $H=(V,E)$ be a hypergraph with  rank $r$ and maximum degree $\Delta$.
 	Moreover, let $H_1, H_2\subseteq H$ be partial hypergraphs of $H$ such that 
	$\skel(H) \cong \skel(H_1)\Box \skel(H_2)$. 
  The numbers $\NrMin$ and $\NrMax$ of non-Cartesian edges in $H_1\boxtimes
  H_2$, $\boxtimes\in \{\strmax, \strmin\}$ 
  can be computed in $O(r^2+|V|\Delta^2)$ time.
	\label{lem:ComputeNr}
\end{lem}
\begin{proof}
	Let $H_1=(V_1, E_1)$ and $H_2=(V_2, E_2)$ be partial hypergraphs of $H$ 
	with 	rank $r_1$, resp., $r_2$ such that $\skel(H)=\skel(H_1)\Box \skel(H_2)$.
	Note, it holds that $r_i \leq r$, $i=1,2$. For the cardinalities $\NrMin$
	and $\NrMax$ we have to compute for pairs of edges $e_1\in E_1$ and
	$e_2\in E_2$ several factorials and for the computation of the Stirling
	number we need in addition values of the form $j^n$. Note, that $m!$,
	resp., $j^n$ can be computed in $O(1)$ time if one knows $(m-1)!$, resp.,
	$j^{n-1}$. Hence, as preprocessing compute first the values $1, 2!,
	\dots, r!$, which can be done in time complexity $O(r)$ and store them
	for later use. Analogously, the complexity for computing the values $j^2,
	\dots, j^{r}$ for a fixed $j\in \{2, \dots, r\}$ is $O(r)$. In that
	manner, we precompute and store the values $2^2,\dots 2^{r}, \dots, r^2,
	\dots, r^{r}$ which takes $O(r^2)$ time. Finally, we store the values of
	the Stirling number, $S_{n,k}$ for $n=1, \dots, r$ and $k=1,\dots, r$.
	Note, $S_{n, k}$ can be computed in $O(1)$ time, whenever $S_{n, k-1}$ is
	known. Hence, for $k, n=1, \dots, r$ the Stirling numbers $S_{n,k}$ can,
	together with the latter preprocessed stored values, be computed in
	$O(r^2)$ time. Therefore, these preprocessing steps have overall time
	complexity of $O(r^2)$.

	After preprocessing and storing the latter mentioned values, one can
	compute the number of non-Cartesian edges in $e_1\strmin e_2$, resp.,
	$e_1\strmax e_2$ in $O(1)$ time, for a fixed pair $e_1\in E_1$ and
	$e_2\in E_2$. These computations are done for all pairs of edges $e_1\in
	E_1$ and $e_2\in E_2$. Hence, we have $|E_1||E_2|$ such computations to
	consider, which take altogether $O(|E_1||E_2|)$ time. Since $|E_i|\leq
	|V_i|\Delta_i$, $i=1,2$ we can conclude that $|E_1||E_2|\leq
	|V_1||V_2|\Delta_1 \Delta_2$. Moreover, by definition of the products, it
	holds that $|V_1||V_2|=|V|$ and since $H_i\subseteq H$ we have 
	$\Delta_i \leq \Delta$, $i=1,2$.
	Therefore, we end in an overall time complexity for computing $\NrMin$
	and $\NrMax$ of $O(r^2+|V|\Delta^2)$.
\end{proof}

\begin{algorithm}[tbp]
\caption{\texttt{Cartesian Skeleton}} 
\label{alg:CartSk}
\begin{algorithmic}[1]
\vspace{1mm}
    \STATE \textbf{INPUT:} A hypergraph $H=(V,E)$;
		\STATE Compute the set $D([H]_2)$ of dispensable edges in $[H]_2$;
    \FOR {every edge $\{x,y\} \in D([H]_2)$}
				\STATE for all edges $e\in E$ with $x,y\in e$ remove $e$ from $E$; 
	  \ENDFOR	
    \STATE \textbf{OUTPUT:} The partial hypergraph $(V, E)$;
\small\normalsize
\end{algorithmic}
\end{algorithm}

\begin{algorithm}[t	]
\caption{\texttt{PFD of thin hypergraphs w.r.t.} $\boxtimes\in \{\protect\strmin, \protect\strmax\}$ } 
\label{alg:PFD}
\begin{algorithmic}[1]
\vspace{1mm}
    \STATE \textbf{INPUT:} A thin hypergraph $H=(V,E)$;
		\STATE Compute the Cartesian skeleton $\skel(H)$ of $H$ with Algorithm \ref{alg:CartSk}; \label{alg:start1}
		\STATE Compute the Cartesian PFD of $\skel(H) = \Box_{i\in I} H_i$ by run of the 
						algorithm of \texttt{Bretto et al.} \cite{BSV-13}\\
    \STATE Assign coordinates $c(v)=(c^v_1, \dots c^v_{|I|})$ w.r.t. $\Box_{i\in I} H_i$ to each vertex $v\in V$; \label{alg:end1}
		\STATE $J \gets I$; 
		\FOR {$k=1,\dots, |I|$} \label{alg:forK}
			\FOR{each $S \subset J$ with $|S|= k$} \label{alg:forS}
				\FOR{$R\in \{S, I\setminus S\}$} \label{alg:forR}
					\STATE Compute $H^R\subseteq H$ with \label{alg:inducedLayers}
	 								$V(H^R) = V(H)$ and \\ 
									$E(H^R) = \{e\in E(H)\mid |p_i(e)|=1, i\in I\setminus R \}$;
				\ENDFOR			\label{alg:endforR}
				\IF{all connected components of $H^S$, resp., $H^{I\setminus S}$  are isomorphic} \label{alg:isom}
					\STATE take one connected component $H_S$ of $H^S$, resp., $H_{I\setminus S}$  of  $H^{I\setminus S}$;
					\IF{all non-Cartesian edges w.r.t. the factorization 
								$H_S \boxtimes H_{I\setminus S}$ are contained in $H$} \label{alg:checkNonCart}
							\STATE save $H_S$ as prime factor; 							
					\ENDIF
				\ENDIF \label{alg:in}
			\ENDFOR \label{alg:endforS}
		\ENDFOR \label{alg:endforK}
		\STATE \textbf{OUTPUT:} The prime factors of $H$;
\small\normalsize
\end{algorithmic}
\end{algorithm}

\begin{thm}
	Algorithm \ref{alg:PFD} computes the prime factors w.r.t. $\boxtimes\in
\{\strmax, \strmin\}$ of a given thin connected simple hypergraph $H=(V,E)$
with maximum degree $\Delta$ and rank $r$ in $O(|V||E|\Delta^6
r^6+|V|^2|E|r)$ time. 
	\label{thm:algoPFD}
\end{thm}
\begin{proof}
	We start to prove the correctness of Algorithm \ref{alg:PFD}. Since
	$H=(V,E)$ is thin, the Cartesian skeleton $\skel(H)$ is uniquely
	determined and the Cartesian prime factors $H_i, i\in I$ of $\skel(H)$
	can be computed with the Algorithm of \texttt{Bretto et al.}
	\cite{BSV-13}. This algorithm returns not only the prime factors of
	$\skel(H)$ but also a coloring of the edges of $\skel(H)$ and thus of the
	edges of $H$. That is, an edge $e\in E$ obtains color $j$ if and only if
	$e\in E(\skel(H))$ and $e$ is an edge of some $H_j$-layer w.r.t.
	$\skel(H)=\Box_{i\in I} H_i$. Hence, this colors the Cartesian edges of
	$H$ w.r.t. the Cartesian PFD of $\skel(H)$ and dispensable edges of $H$
	obtain no color. Based on $\skel(H)$ one can compute the coordinates in
	the following way. One first computes $[\skel(H)]_2$ and coordinatize the
	vertices of $V([\skel(H)]_2)=V$ as proposed in \cite[page
	280]{Hammack:2011a} w.r.t. to the product coloring given by $\Box_{i\in
	I} H_i$. Note, then for all edges $e=\{x,y\}\in E([\skel(H)]_2)$ holds
	$|p_i(e)|=2$ if and only if the coordinates of $x$ and $y$ differ in the
	$i$-th coordinate and the other coordinates are identical. To prove that
	this is a valid coordinatization of $\skel(H)$ one has to show, that for
	all edges $e\in E(\skel(H))$ holds that $|p_i(e)|>1$ if and only if for
	all $x,y\in e$ holds that $x$ and $y$, differ in the $i$-th coordinate
	and the other coordinates are identical. Let $e\in E(\skel(H))$ be an
	arbitrary edge. This edge forms a complete subgraph in the 2-section $[\skel(H)]_2$.
	However, complete subgraphs must be contained entirely in one of the $H_i$-layers of
	$[\skel(H)]_2$, as complete graphs are so-called S-prime graphs, see e.g.
	\cite{Hel-12, HGS-09}. From this we can conclude that the computed
	coordinates of vertices in $[\skel(H)]_2$ give a valid coordinatization
	of the vertices in $\skel(H)$. 

	Now, consider Line \ref{alg:forK}-\ref{alg:endforK}.
	We finally have to examine which ``combination'' of the proposed
	Cartesian prime factors are prime factors w.r.t. $\boxtimes$ (Line
	\ref{alg:forK}-\ref{alg:endforK}). For this, we search for the minimal
	subsets $S$ of $I$ such that the subgraph $H_S$ and $H_{I\setminus S}$,
	where $H_S$ is one connected component of $H^S$ and $H_{I\setminus S}$ is
	one connected component of $H^{I\setminus S}$, correspond to layers of a
	factor of $H$ w.r.t. $H_S \boxtimes H_{I\setminus S}$. We continue to
	check whether all connected components of $H^S$, resp., $H^{I\setminus
	S}$ are isomorphic and if so, we test whether all non-Cartesian edges w.r.t.
	the factorization $H_S \boxtimes H_{I\setminus S}$ are present. If this
	is the case, $H_S$ is saved as prime factor of $H$ w.r.t. $\boxtimes$.
	Reasoning exactly as in the proof for graphs in \cite[Chapter
	24.3]{Hammack:2011a} together with the preceding results, we conclude the
	correctness of this part in Line \ref{alg:forK}-\ref{alg:endforK}.

  We are now concerned with the time complexity. Note, since we assumed the
  hypergraph $H=(V,E)$ to be connected we can conclude that $[H]_2$ has at
  least $|V|-1$ edges. Moreover, the number of edges in $[H]_2$ does not
  exceed $|E|r^2$ and therefore we can conclude that $O(|V|^2) \subseteq
  O(|V||E|r^2)$. Furthermore, we will make in addition frequent use of the
  fact that $|E|\leq |V|\Delta$. Now, consider Line
  \ref{alg:start1}-\ref{alg:end1}. Lemma \ref{lem:timeCartSk-hypergraph}
  implies that the Cartesian skeleton can be computed in $O(|E|^2 r^4)
  \subseteq O(|V|^2\Delta^2 r^4) \subseteq O(|V||E|\Delta^2 r^6)$ time and
  by Theorem \ref{thm:bretto} we have that the PFD of $\skel(H)$ can be
  computed in $O(|V||E|\Delta^6 r^6)$ time. For the computation of the
  coordinates we use the 2-section $[\skel(H)]_2 $ as described in the
  previous part of this proof. Note, $[\skel(H)]_2 $ has at most $|E|r^2$
  edges and the coordinates can therefore be computed in $O(|E|r^2)$, see
  \cite[Chapter 23.3]{Hammack:2011a}. Hence, the overall time complexity of
  the steps in Line \ref{alg:start1}-\ref{alg:end1} is
  $O(|V||E|\Delta^6r^6)$.

  Consider now Line \ref{alg:forK}-\ref{alg:endforK}. Clearly, each $H^R$
  can be computed in $O(|E|r)$ time. For finding the connected components
  of $H^R$ in Line \ref{alg:isom} one can use its 2-section
  $[H^R]_2=(V,E')$ and apply the classical breadth-first search to it,
  which has time complexity $O(|E'| + |V|)$. Let $\Delta'$ be the maximum
  degree of $[H^R]_2$ which is bounded by $\Delta r$. Hence, we can
  determine the connected components of $H^R$ in time complexity $O(|E'| +
  |V|) \subseteq O(|V|\Delta')\subseteq O(|V|\Delta r)$. Moreover, in Line
  \ref{alg:isom} we have to perform an isomorphism test for a fixed
  bijection given by the coordinates which takes $O(|E|r)$ time. This test
  must be done for each of the connected components of $H^R$ which are at
  most $|V|$. Hence, the latter task has time complexity $O(|V||E|r)$.
  Taken together the preceding considerations and since $\Delta \leq |E|$
  we can conclude that Line \ref{alg:isom} can be performed in $O(|V|\Delta
  r+|V||E|r)=O(|V||E|r)$ time. To test whether all non-Cartesian edges
  w.r.t. $H_S \boxtimes H_{I\setminus S}$ are contained in $H$ (Line
  \ref{alg:checkNonCart}) we examine whether putative non-Cartesian edges $
  e\in E(H) \setminus E(H_S \Box H_{I\setminus S})$ are valid non-Cartesian
  edges, that is, we prove if the projection properties for these edges
  into the factors fulfill the condition $(ii)$ in the definition of edges
  in $H_S \boxtimes H_{I\setminus S}$ and count them, if valid. If the
  counted number is identical to $\NrMin$, resp., $\NrMax$ we are done.
  Since the coordinates are given, the projections can be computed in
  $O(|E|r)$ time. The computation of $\NrMin$, resp., $\NrMax$ has time
  complexity $O(r^2+|V|\Delta^2)$ (Lemma \ref{lem:ComputeNr}). Thus, Line
  \ref{alg:checkNonCart} can be performed in $O(|E|r+r^2+|V|\Delta^2)$
  time. Taken together all the single tasks in Line
  \ref{alg:forR}-\ref{alg:in} we end up in a time complexity
  $O(|E|r+|V||E|r+|V|\Delta^2+r^2)=O(|V||E|r+|V|\Delta^2+r^2)$. Assume all
  these tasks are done for each of the the $2^{|I|}$ subsets of $I$. Since
  $|I|$ is the number of factors of $\skel(H)$ and thus, is bounded by
  $\log_2 (|V|)$ we have at most $|V|$ subsets of $I$. To summarize, the
  total complexity of Line \ref{alg:forK}-\ref{alg:endforK} is
  $O(|V|^2|E|r+|V|^2\Delta^2+|V|r^2)$. Since $H$ is assumed to be
  connected we can conclude that $O(|V|^2)\subseteq O(|V||E| r^2)$ and
  hence, the complexity of Line \ref{alg:forK}-\ref{alg:endforK} is
  $O(|V|^2|E|r+|V||E|\Delta^2 r^2 + |V|r^2)$. 
	
	Taken together the preceding results we can infer that 
	Algorithm \ref{alg:PFD} has time complexity 
  $O(|V||E|\Delta^6	r^6 + |V|^2|E|r+|V||E|\Delta^2 r^2 + |V|r^2)$, that is, 
  $O(|V||E|\Delta^6	r^6+  |V|^2|E|r)$.
\end{proof}

\begin{cor}
	Algorithm \ref{alg:PFD} computes the prime factors w.r.t. $\boxtimes\in
	\{\strmax, \strmin\}$ of a given thin connected simple hypergraph
	$H=(V,E)$ with bounded degree and bounded rank in $O(|V|^2 |E|)$ time.
\end{cor}

\section*{Acknowledgment}
This work was supported in part by the \emph{Deutsche
Forschungsgemeinschaft} within the EUROCORES Programme EUROGIGA (project
GReGAS) of the European Science Foundation.

\bibliographystyle{plain}
\bibliography{Outline}

\end{document}